\RequirePackage[l2tabu,orthodox]{nag}
\documentclass
[11pt,letterpaper]
{article} 

\usepackage[notes=true,later=false,camera=false]{dtrt}
\usepackage[utf8]{inputenc}
\usepackage{etex}
\usepackage{ stmaryrd }
\usepackage{xspace,enumerate}
\usepackage[T1]{fontenc}
\usepackage[full]{textcomp}
\usepackage[american]{babel}
\usepackage{mathtools}

\usepackage{amsthm}
\usepackage{thmtools}
\usepackage{thm-restate}

\usepackage{empheq}

\hypersetup{
colorlinks=true,
urlcolor=Cerulean,
linkcolor=RoyalBlue,
citecolor=OliveGreen,
linktocpage=true,
}
\renewcommand*\backref[1]{\ifx#1\relax \else (pg. #1) \fi}

\usepackage[capitalise,nameinlink]{cleveref}
\crefname{lemma}{Lemma}{Lemmas}
\crefname{fact}{Fact}{Facts}

\crefname{theorem}{Theorem}{Theorems}
\crefname{mtheorem}{Theorem}{Theorems}
\crefname{itheorem}{Theorem}{Theorems}
\crefname{corollary}{Corollary}{Corollaries}
\crefname{claim}{Claim}{Claims}
\crefname{example}{Example}{Examples}
\crefname{algorithm}{Algorithm}{Algorithms}
\crefname{problem}{Problem}{Problems}
\crefname{definition}{Definition}{Definitions}
\crefname{equation}{Eq.}{Eq.}
\crefname{strategy}{Strategy}{Strategies}

\usepackage{paralist}
\usepackage{turnstile}
\usepackage[framemethod=TikZ]{mdframed}
\mdfsetup{frametitlealignment=\center}
\usepackage{tikz}
\usepackage{caption}
\DeclareCaptionType{Algorithm}
\usepackage{newfloat}

\newtheorem{theorem}{Theorem}[section]
\newtheorem{lemma}[theorem]{Lemma}
\newtheorem*{lemma*}{Lemma}

\newtheorem{proposition}[theorem]{Proposition}
\newtheorem{fact}[theorem]{Fact}

\theoremstyle{definition}
\newtheorem{definition}[theorem]{Definition}
\newtheorem*{definition*}{Definition}
\newtheorem{remark}[theorem]{Remark}

\newtheorem{algorithm}{Algorithm}
\newtheorem{algorithm-thm}[theorem]{Algorithm}

\usepackage[
letterpaper,
top=1.2in,
bottom=1.2in,
left=1in,
right=1in]{geometry}
\usepackage{newpxtext} 
\usepackage{textcomp} 
\usepackage{mathpazo}
\usepackage[scr=rsfso]{mathalfa}
\usepackage{bm} 
\let\mathbb\varmathbb
\usepackage{microtype}

\usepackage{footnotebackref}

\allowdisplaybreaks
\newcommand{\FormatAuthor}[3]{
\begin{tabular}{c}
#1 \\ {\small\texttt{#2}} \\ {\small #3}
\end{tabular}
}

\newcommand{\R}{{\mathbb R}}
\newcommand{\N}{{\mathbb N}}

\newcommand{\eps}{\varepsilon}
\newcommand{\defeq}{\coloneqq}
\newcommand{\F}{{\mathbb F}}
\newcommand{\cF}{{\mathcal F}}
\newcommand{\E}{{\mathbb E}}

\newcommand{\cV}{{\mathcal V}}

\newcommand{\ip}[1]{\langle #1 \rangle}

\newcommand{\Fits}{\{\pm 1\}}
\newcommand{\Var}{\mathrm{Var}}

\newcommand{\pE}{\tilde{\E}}
\newcommand{\cH}{\mathcal H}

\newcommand{\cB}{\mathcal B}

\newcommand{\cS}{\mathcal S}
\newcommand{\poly}{\mathrm{poly}}
\newcommand{\val}{\mathrm{val}}
\newcommand{\coloneqq}{:=}
\newcommand{\bias}{\mathrm{bias}}

\newcommand{\mper}{\,.}
\newcommand{\mcom}{\,,}

\newcommand{\aio}{a_{s_{1}}}

\newcommand{\cD}{\mathcal D}

\renewcommand{\emptyset}{\varnothing}
\renewcommand{\geq}{\geqslant}
\renewcommand{\leq}{\leqslant}
\renewcommand{\epsilon}{\varepsilon}
\newcommand{\pldist}{\cD_{\mathrm{pl}}}
\newcommand{\nulldist}{\cD_{\mathrm{null}}}

\newcommand{\bjk}{B_{k,v}^{(j)}}
\newcommand{\wvx}{w_{v,x}}

\usepackage{macros}
\begin{document}

\title{Strongly Refuting Semirandom Linear Systems in Subexponential Time}
 \author{
 \begin{tabular}{cc}
 \FormatAuthor{Pravesh K. Kothari}{kothari@cs.princeton.edu}{Princeton University} 
  \FormatAuthor{Andrew D. Lin\thanks{This material is based upon work supported by a Princeton AI Lab Seed Grant.}}{andrewlin@princeton.edu}{Princeton University} 
  \FormatAuthor{Peter Manohar\thanks{This material is based upon work supported by the National Science Foundation
under Grant No.\ DMS-2424441.}}{pmanohar@ias.edu}{The Institute for Advanced Study} 
  \end{tabular}
  }
\maketitle
\vspace{-1em}

\begin{abstract}
In this paper, we consider the problem of refuting $\F_2$-linear equations with random right-hand sides. Formally, we give a \emph{sub-exponential} $2^{O(n/\log n)}$-time randomized algorithm that takes as input an arbitrary $m \times n$ matrix $A$ and a uniformly random vector $b \in \F_2^m$, and outputs a witness showing that no assignment satisfies more than a $\frac{1}{2}+\epsilon$ fraction of the equations provided that $m \geq 2^{O(n/\log n)}$. The setting above is the \emph{semirandom} refutation variant of the famous work~\cite{BKW03} that gives a $2^{O(n/\log n)}$-time search algorithm for the learning parity with noise (LPN) problem with $m \geq 2^{O(n/\log n)}$ equations.

Building on the search algorithm of~\cite{Lyub05}, we also give a $2^{O(n/\log \log n)}$-time refutation algorithm that succeeds with only $m \geq n^{1 + \gamma}$ equations, for a small constant $\gamma$. 
Finally, we prove that our algorithm is \emph{not} captured by the sum-of-squares hierarchy by proving a degree-$\Omega(n)$ sum-of-squares lower bound, showing that ``\cite{BKW03}-style'' algorithms achieve better runtime than can be done under sum-of-squares. We thus obtain a natural example of a noise-tolerant signal recovery problem that exhibits a nontrivial gap between the performance of efficient algorithms and that of those based on the sum-of-squares hierarchy.
\end{abstract}
\vspace*{\fill}
\paragraph{Statement of AI Use.}
The main content of this paper, e.g., the central question, high-level proof strategy, and technical content, is due solely to the human authors. We have only used AI models for literature review and proofreading. This statement is intended solely to describe the role of AI in the research
process and should not be read as an endorsement of such systems or of the companies that produce them.

\thispagestyle{empty}
\clearpage
 \microtypesetup{protrusion=false}
  \tableofcontents{}
  \microtypesetup{protrusion=true}

\thispagestyle{empty}
\clearpage

\pagestyle{plain}
\setcounter{page}{1}
\section{Introduction}
\label{sec:intro}
In the \emph{Learning Parity with Noise} (LPN) problem, we are given as input a system of inhomogeneous linear equations in $n$ variables and $m$ equations of the form $Ax = b$, where $A \in \F_2^{m \times n}$, $x \in \F_2^n$ are the variables, and $b \in \F_2^m$, and the goal is to distinguish between the following two distributions $\nulldist$ and $\pldist$. In $\nulldist$, $A$ and $b$ are drawn uniformly at random; in $\pldist$, $A$ is drawn uniformly at random, and $b$ is chosen so that $b = Ax^* + e$, where $x^* \in \F_2^n$ is a hidden (or planted) assignment and $e \in \F_2^m$ is sampled by setting each $e_i = 1$ with probability $1/2 - \eps$, independently. The hardness of LPN was first used as a cryptographic hardness assumption in \cite{BlumFKL94}, where it was used to construct pseudorandom generators, and has since become a common hardness assumption. Further applications include secret-key encryption \cite{GilbertRS08}, message authentication \cite{HopperB01,JuelsW05}, and (non-interactive) zero-knowledge proofs \cite{JainKPT12, DaoJJ24}.

The LPN distinguishing task can be trivially solved in $2^n$ time by enumerating over all $x \in \F_2^n$; the strength of it as a cryptographic assumption therefore depends on whether more efficient algorithms exist. A breakthrough work of \cite{BKW03} gave an algorithm for the distinguishing task that runs in $\poly(m)$ time given $m=2^{O(n/\log n)}$ equations, i.e., slightly subexponential time and samples. In fact, the algorithm additionally solves the harder ``search task'': when the equations are drawn from $\pldist$ with the planted assignment $x^*$, the algorithm recovers $x^*$ with high probability. Furthermore, the algorithm succeeds in the harder ``semirandom'' setting, where the matrix $A$ is permitted to be arbitrary, while the distribution over the ``right-hand sides'' $b$ remains the same.\footnote{We note that in this "semirandom" regime it is not always possible to recover $x^*$ uniquely, as for example, it is possible that the $j$-th column of $A$ is $0^m$ and thus we have no information about $x_j$. Instead, the algorithm recovers an assignment $\hat{x}$ that is consistent with an $\frac{1}{2}+\Omega(\eps)$-fraction of the samples.}
A follow-up work by \cite{Lyub05} extended the  \cite{BKW03}~algorithm in the random setting to only require $m=n^{1+\gamma}$ equations, for any small constant $\gamma$, at the cost of increasing the runtime of the algorithm to $2^{O(n/\log\log n)}$ time. The key idea in~\cite{Lyub05} is to use the original $m$ equations to generate $2^{O(n/\log \log n)}$ synthetic equations and then run the  \cite{BKW03}~algorithm on the synthetic equations. These two results remain the state-of-the-art after over 20 years.

While the algorithms of~\cite{BKW03,Lyub05} for LPN solve both the distinguishing task and search task, for average-case problems there are typically three variants of the problem that are studied: distinguishing, search, and refutation. The distinguishing task asks an algorithm to be able to distinguish between the two distributions $\nulldist$ and $\pldist$, while in the search task, the goal is to find the hidden assignment $x^*$ given a sample from $\pldist$. In the refutation task, however, the goal is to output a certificate that there is no assignment $\hat{x} \in \F_2^n$ that satisfies more than $1/2 + \eps$-fraction of the equations, given a sample from $\nulldist$.

To the best of our knowledge, there has so far been no work on studying the refutation problem for LPN and no nontrivial algorithm is known. However, the refutation problem for the sparse version of the problem, sparse LPN (also commonly called $k$-XOR), has been extensively studied, both on the side of algorithms~\cite{GoerdtL03, CojaGL07,AllenOW15,BarakM16,RaghavendraRS17,AbascalGK21,GuruswamiKM22,dOrsiT23} and lower bounds in restricted computational models~\cite{Grigoriev01,Sch08, BCK15, KothariMOW17}.
Like its dense counterpart, sparse LPN has many applications in cryptography, beginning with 
\cite{Alekhnovich03}, where it was used to construct a public-key cryptographic scheme, and has since been used in signature schemes \cite{IshaiKOS08}, indistinguishability obfuscation ($i\mathcal{O}$) \cite{RagavanVV25}, and multi-party homomorphic secret sharing \cite{DaoIJL23}, among other applications.

This state of affairs raises the following two natural questions: (1) are there algorithms for the refutation version of LPN that are analogues of~\cite{BKW03,Lyub05}, and (2) given that to date, there has been no asymptotic improvements to the algorithms of~\cite{BKW03,Lyub05}, can we provide some evidence that these algorithms may be optimal?

In this work, we resolve the first question by giving the first subexponential-time refutation algorithm for random and semirandom linear systems. As we have mentioned, the semirandom case is the more challenging setting where the equation matrix $A$ is arbitrary, and we require that our algorithm succeeds for all choices of $A$. Our algorithm's runtime and sample complexity matches the guarantees of the search algorithms in \cite{BKW03} and \cite{Lyub05}. Furthermore, even in the setting where $m=n^{1+\gamma}$, our refutation algorithm works for semirandom linear systems, unlike the search algorithm of \cite{Lyub05}, which relies heavily on the randomness of $A$.

A natural approach to the second question is to leverage the approach taken in the study of sparse LPN and prove lower bounds for the \emph{refutation} problem in restricted computational models, such as, e.g., the sum-of-squares semidefinite programming hierarchy. As a second result, we show that unlike sparse LPN, where the best-known algorithms are all ``captured'' by the powerful sum-of-squares semidefinite programming hierarchy, our refutation algorithms for (dense) LPN are \emph{not} captured by the sum-of-squares hierarchy. We do this by proving a degree-$\Omega(n)$ sum-of-squares lower bound for refuting semirandom LPN instances. As ``degree'' in the sum-of-squares hierarchy roughly corresponds to runtime, this is typically interpreted as a $2^{\Omega(n)}$-time lower bound, a lower bound which is violated by our algorithm. Thus, any approach to proving lower bounds for dense LPN will require going beyond the standard sum-of-squares framework.

\subsection{Our results}
Before we formally state our results, we introduce a few definitions. We start by defining random and semirandom linear systems.
\begin{definition}
    A \emph{linear system in $\F_2^n$} with $m$ equations is a pair $(A,b)\in \F_2^{m\times n}\times \F_2^m$, which represents the system of inhomogeneous linear equations given by $Ax = b$. We define the \emph{value} $\val(A,b)$ of $(A,b)$ to be the maximum over all $x\in \F_2^n$ of the fraction of $i\in [m]$ such that $\ip{a_i,x}=b_i$, where $a_i$ denotes the $i$-th row of $A$, i.e., it is the maximum fraction of satisfiable equations over all assignments. The  \emph{bias} of $(A,b)$ is defined as $\bias(A,b)\coloneqq \val(A,b)-\frac{1}{2}$. We say that $(A,b)$ is a \emph{random linear system} if $A$ and $b$ are uniform and independent random elements of $\F_2^{m\times n}$ and $\F_2^m$, respectively, and that $(A,b)$ is a \emph{semirandom linear system} if $b$ is a uniformly random element of $\F_2^m$ independent of $A$, which can be arbitrary.
\end{definition}

As our first result, we give an algorithm for refuting semirandom linear systems with parameters similar to the search algorithm of~\cite{BKW03}.

\begin{restatable}{theorem}{semirandomfull}
    \label{thm:semirandomfull}
    Let $m$ be an integer such that $m\geq2^{3n/\log n}$. There is an algorithm that takes as input a linear system $(A,b)$ with $m$ equations and in  $\poly(m)$ time outputs a quantity $\cV\in [\frac12,1]$ such that:
    \begin{itemize}
        \item $\val(A,b)\leq \cV$,
        \item If $(A,b)$ is a semirandom linear system, then for all $2^{-n^{1/20}}\leq \eps\leq\frac12$ we have $\cV\leq \frac{1}{2}+\eps$ with high probability. 
    \end{itemize}
\end{restatable}

Note that \cref{thm:semirandomfull} immediately gives an algorithm for random linear systems as well, as the case when $A$ is random is a special case of a semirandom linear system.

For our second result, we extend our algorithm so that we only require a superlinear number of equations, at the cost of an increase in runtime. This gives a refutation algorithm with parameters similar to the search algorithm of~\cite{Lyub05}.
\begin{restatable}{theorem}{semirandomgeneral}
    \label{thm:semirandomgeneral}
        Let $\gamma > 0$ and let $m=n^{1+\gamma}$. There is an algorithm taking in a linear system $(A,b)$ with $m$ equations and in $2^{\Theta(n/\log \log n)}$ time outputs a quantity $\cV\in (\frac{1}{2},1]$ such that:
    \begin{itemize}
        \item $\val(A,b)\leq \cV$,
        \item If $(A,b)$ is a semirandom linear system, then for any $\frac{3}{\log^{4/3}n}\leq\eps\leq \frac12$, we have $\cV\leq \frac{1}{2}+\eps$ with high probability.
    \end{itemize}
\end{restatable}

Finally, to demonstrate that our algorithm is not captured by the sum-of-squares hierarchy, we prove an $\Omega(n)$-degree lower bound for sufficiently dense linear systems in the sum-of-squares hierarchy. 
\begin{theorem}[Informal \cref{thm:soslowerbound}]
\label{thm:infsoslowerbound}
    Let $(A,b)$ be a linear system with $m=2^{\Theta(n/\log n)}$ equations where (1) each row $a_i$ is sampled by independently setting each coordinate $(a_{i})_j=1$ with probability $q=\Omega(\frac{1}{\log n})$, and (2) $b \gets \F_2^m$ uniformly at random. Then, with high probability over the draw of $(A,b)$, the following hold:
        \begin{itemize}
        \item $\val(A,b)\leq \frac{1}{2} + o(1)$;
        \item The canonical degree-$\Omega(n)$ sum-of-squares relaxation has value $1$, i.e., it fails to refute $(A,b)$.
    \end{itemize}
    \end{theorem}
We note that the distribution over $(A,b)$ in \cref{thm:infsoslowerbound} is \emph{not} the standard LPN distribution. This is because there is a technical issue with encoding a system of dense linear equations as a sum-of-squares optimization problem, as the latter requires encoding each linear equation over $\F_2$ into a polynomial equation over $\R$. The standard method for doing this converts a linear equation $\ip{a_i,x}=b_i$ into the polynomial equation $\prod_{j \in \supp(a_i)}y_j = (-1)^{b_i}$, where $y_j \defeq (-1)^{x_j}$ and $\supp(a_i) \defeq \{j : (a_i)_j = 1\}$. This requires polynomials of degree $\approx n/2$ when the $a_i$'s are random. Thus, the degree of the sum-of-squares hierarchy must be at least $n/2$ for the sum-of-squares relaxation to even be \emph{defined}. In \cref{thm:infsoslowerbound}, we circumvent this issue by proving a lower bound for the specific case where each equation $a_i$ is sampled randomly with support of size $n/\log n$ in expectation. While this does differ from the standard LPN distribution,\footnote{The typical intuition is that sparser equations should be ``easier'', so one might expect this distribution to in fact be easier than the standard LPN distribution, thus making the lower bound stronger. However, formally this is not known.} 
where $A$ is drawn uniformly at random, we note that the algorithm in \cref{thm:semirandomfull} succeeds for semirandom systems, i.e., those with arbitrary $A$, so in particular it also succeeds on this input distribution. Hence, this is sufficient to assert that \cref{thm:semirandomfull} is not captured by the sum-of-squares hierarchy.

\section{Proof Overview}
\label{sec:overview}
In this section, we give an overview of the proofs of \cref{thm:semirandomfull,thm:semirandomgeneral,thm:infsoslowerbound}.

\subsection{The search algorithm of \cite{BKW03}}
To begin, we will first give a summary of the search algorithm of~\cite{BKW03} for random instances of LPN. Let us explain the main idea of the BKW algorithm first. Let's first start with a satisfiable system of linear equations $\{ \langle a_i, x \rangle = b_i\}$. Say we find a set of $a_i$s that add up to $e_j$. Then, adding the corresponding equations gives us the value of $x_j$ -- the $j$th bit of the hidden solution. Now, when each of the equations are erroneous with probability $1/2-\epsilon$, the linear combination has correlation $(2\epsilon)^t$ (instead of $1$) with $x_j$. In particular, the amount of correlation drops exponentially in the \emph{sparsity} of the linear combination. It is thus beneficial to find a linear combination that adds up to $e_j$ while being as sparse as possible. 

The key subroutine in the algorithm shows precisely such a result. Given input $m=2^{\Omega(n/\log n)}$ random vectors $a_1, \cdots, a_m$, it finds a set $S$ of size $\ll \frac{n}{\log n}$ (thus, significantly smaller than the trivial bound of $n$) such that $\sum_{i\in S}a_i=e_j$, where $e_j$ is any standard basis vector. Since
$\ip{a_i,x^*}=b_i$ independently with probability $\frac12+\eps$ for all $i$, it follows that $\sum_{i \in S} b_i$ is equal to $x^*_1$ with probability $\frac{1}{2}+\frac12(2\eps)^{t}$, where $t = \abs{S}$. Therefore, if we can find $\gg 2^{\Omega(t)}$ total such disjoint sets $S$, each of size $\leq t$, then by a Chernoff Bound a majority vote recovers $x^*_1$. 

We can construct such a set $S$ by repeatedly combining the $a_i$'s to ``zero out'' blocks of coordinates. More precisely, their algorithm proceeds as follows. First, divide the coordinates $1, \cdots, n$ into $r=\frac{1}{2}\log n$ blocks of size $s=\frac{2n}{\log n}$. 
Given a subset $\{a_1, \cdots, a_M\}$ of $M=2^{\Omega(n/\log n)}$  vectors, where $M\leq m$, start at the rightmost block consisting of the last $s$ coordinates, and bucket the vectors $a_i$ based on their last $s$ coordinates. Then, pick a random vector in each bucket which contains at least $2$ vectors, and add it to all other vectors in the bucket, which zeroes out the last $s$ coordinates of all those vectors. We are left with a set of at least $m-2^s$ new vectors whose last $s$ coordinates are $0$. Note that since all coordinates of the random vectors are independent, the vectors produced from this process are still uniformly random in their first $n-s$ coordinates. We can then repeat the above on the second block from the right, and so on, until only the last block is nonzero. The last block will contain uniformly random vectors in $\F_2^s$, so for sufficiently large $m=2^{\Theta(n/\log n)}$, it will contain the vector $e_1$ with high probability. We have thus written $e_1$ as the sum of at most $2^{r-1} \leq \frac{\sqrt{n}}{2}$ different $a_1, \cdots, a_m$, so the sum of the corresponding RHS of these equations is equal to $x^*_1$ with probability at least $\frac{1}{2}+\frac12(2\eps)^{\sqrt{n}/2}$.

Thus, we can recover $x^*_1$ with high probability by dividing the original set of $m$ vectors into $2^{O(n^{1/2})}$ sets of $M = 2^{O(n/\log n)}$ vectors, and using each set of $M$ vectors to find such a set $S$ for $e_1$ and then taking the majority vote. Repeating the strategy for each coordinate (noting that we can either zero out blocks out of order or permute the coordinates of the vectors) recovers $x_2^*, \cdots, x^*_n$, and thus all of $x^*$. This requires a total of only $n \cdot 2^{O(\sqrt{n})} \cdot M=2^{\Theta(n/\log n)}$ samples, so the above algorithm (which is $\poly(m)$ time) runs in $2^{\Theta(n/\log n)}$ time. 

\subsection{Refutation via BKW-style bucketing: challenges and our approach}
Can we extend such an idea to the refutation problem? There's a natural variant of the above procedure that one may try: generate sets $S$ of $|S|=2^r=\sqrt{n}$ vectors that sum to the zero vector. Let $\cS$ be the set of these $S$'s. Then, if there exists some $x^*$ such that $\ip{a_i,x^*}=b_i$ for a $\geq\frac12+\eps$ fraction of $i$, then we would expect that for $\geq\frac12+\frac12(2\eps)^{\sqrt{n}}$ fraction of $S \in \cS$, $\sum_{i\in S}b_i=\ip{\sum_{i\in S}a_i, x^*}=0$. Hence, if no such $x^*$ exists, we should expect that $\sum_{i\in S}b_i = 0$ for $\leq \frac{1}{2} + O(\eps^{\sqrt{n}})$ fraction of $S \in \cS$. 

Unfortunately, such a procedure does not constitute a \emph{certificate} that there is no good solution for the original system of equations. Indeed, it is not hard to show that there exist linear systems $(A,b)$ that satisfy the above condition but also have $\val(A,b) \geq \frac{1}{2} + \eps$.\footnote{Let $a_i=e_1$ for all $i\in [m]$, $b_1, \cdots, b_{m-1}=0$, and $b_{m}=1$. Pick $S_i=\{m\}\cup P_i$ for all $i\in [M/2]$, where the $P_1, \cdots, P_{M/2}$ are distinct elements of ${[m-1]\choose \sqrt{n}-1}$ in some order, and distinct $S_{M/2+1}, \cdots S_M\in {[m-1]\choose \sqrt{n}}$. Then we have $\sum_{j\in S_i}a_j=\mathbf0$ for all $i$, and $\sum_{j\in S_i}b_j=1$ for all $i\leq M/2$ and $0$ for all $i>M/2$.
} At a high level, the issue is that the algorithm takes the initial system $(A,b)$ and produces a new linear system $(A',b')$ where there is no obvious relation between $\val(A,b)$ and $\val(A',b')$. So, even if we can argue that $\val(A',b')$ is small, it is not clear how to conclude that $\val(A,b)$ must also be small.

As a first step to obtaining a refutation algorithm, we make the observation that if \emph{all of the buckets are of equal size}, and we \emph{combine all pairs of equations in each bucket}, then we will obtain a linear system $(A',b')$ with $m'= m^2/2^s$ equations,
each the sum of two original equations, such that if $\bias(A,b)\geq \eps$, then $\bias(A',b')\geq\Omega(\eps^2)$. That is, if $\val(A,b) \geq \frac{1}{2} + \eps$, then $\val(A',b') \geq \frac{1}{2} + \Omega(\eps^2)$. This follows from the Cauchy--Schwarz inequality: letting $S_v$ be the set of all $i$ such that the last block of $a_i$ is equal to $v\in \F_2^{s}$, and viewing a linear equation $\ip{a_i,x}=b_i$ as the function $p_i(x)=(-1)^{\ip{a_i, x} - b_i}$, we see that 
\begin{flalign*}
    \bias(A,b) &=\max_{x \in \F_2^n} \frac{1}{m}\sum_{v \in \F_2^{s}}\left(\sum_{i\in S_v}p_i(x)\right) \\
    &\leq \max_{x \in \F_2^n} \sqrt{ 2^s \cdot \left( \frac{1}{m^2}\sum_{v \in \F_2^s}\left(\sum_{i\in S_v}p_{i}(x)\right)^2\right)} \mper
\end{flalign*}
Note that we have
\begin{flalign*}
\bias(A',b') &= \max_{x \in \F_2^n} \frac{1}{\sum_{v \in \F_2^s} \abs{S_v}^2} \sum_{v \in \F_2^s} \left(\sum_{i\in S_v}p_{i}(x)\right)^2 \mcom
\end{flalign*}
and if the $S_v$'s are balanced, then $\abs{S_v} = m/2^s$, which implies that $\sum_{v \in \F_2^s} \abs{S_v}^2 = m^2/2^s$, and hence $\sqrt{\bias(A',b')} \geq \bias(A,b)$.

Therefore, if we again continue this process for all $r$ blocks, we will eventually produce an instance $(A'',b'')$ where (1) $\bias(A'',b'')^{1/2^r} \geq \bias(A,b)$, and (2) all equations in $A''$ have $a''_i = 0^n$, so that $\bias(A'',b'')$ is simply the absolute value of the fraction of $1$'s minus the fraction of $0$'s in $b''$, and this is indeed a certificate of refutation for the original instance $(A,b)$.

However, there is a problem: at each step of this process, the number of equations roughly grows from $m$ to $m^2/2^s$, which means that after $r$ steps, we have $m^{2^r} / 2^{sr}$ equations. Recall that $s = \frac{2 n}{ \log n}$, $r s = n$, and $m = 2^{O(s)}$, so that $m^{2^r} / 2^{sr} = 2^{O(n^{3/2} / \log n)}$, which is larger than $2^{O(n)}$. Hence, this algorithm as stated is worse than the naive exponential-time algorithm that tries all assignments $x \in \F_2^n$.

We can control the growth of the number of equations via the following simple idea: rather than working with the original buckets at each step, we instead split the buckets into sub-buckets of size $2^{O(n/\log^2n)}$, and apply Cauchy--Schwarz only within a sub-bucket. That is, we only ``combine all pairs'' within each sub-bucket. This causes the number of vectors $m$ to grow by only a factor of $2^{O(n/\log^2n)}$ at each step, so final number of vectors is $m \cdot 2^{r \cdot O(n/\log^2 n)}$, which is $2^{O(n/\log n)}$, and thus our runtime is still $2^{O(n/\log n)}$. We note that to prove our algorithm outputs with high probability a nontrivial bound $\cV$ on the value, we must show there are not many ``trivial even covers'', i.e. $2^r$ sparse combinations of vectors, such that every vector appears an even number of times in the combination. We present our full algorithm and analysis in \Cref{sec:semirandom}.

\parhead{Extension to the~\cite{Lyub05} regime.} The algorithm of~\cite{Lyub05} works by starting with $m=n^{1+\gamma}$ equations and generating synthetic samples by taking random sparse combinations of the original equations. It samples $M=2^{O(n/\log \log n)}$ random $q=\Theta(\frac{n}{\log n})$-sparse combinations of vectors, where $q$ is chosen to be sufficiently large to ensure the random sampling produces equations which are approximately independent and random. In particular, their analysis requires $q$ to be chosen such that there are ${m\choose q}\geq 2^{2n}$ possible synthetic equations to sample from. Also, it requires the original $m$ vectors to be uniformly random, unlike \cite{BKW03} which works for the semirandom case.

We extend our algorithm to this setting by producing a deterministic sample generation process to generate sufficiently many synthetic equations to run the refutation algorithm. Our sample generation process differs from that of \cite{Lyub05} in a few crucial ways: to relate the value of the synthetic system to the original system, we must ensure that we generate samples in a way such that if original system has value $\frac{1}{2}+\eps$, then the resulting synthetic system has value $\frac{1}{2}+\Omega(\eps^q)$. At the same time, we must ensure that the set of $q$-sparse combinations we choose must not contain many length $2\sqrt{\log n}$ even covers to ensure we output a nontrivial bound with high probability. This is because we will combine up to $\sqrt{\log n}$ synthetic equations together when we run the refutation algorithm, and if two of these combine to form an even cover, then the RHS of the equations will be dependent, weakening our concentration bounds. We note that given an unbalanced $(\sqrt{\log n},\frac12-\Omega(1))$-lossless expander with left-degree $O(\frac{n}{\log n})=O(\frac{\log M\log\log\log M}{\log\log M})$, one can simply choose the left-neighborhood sets as our sparse combinations; however, the best known explicit constructions \cite{GuruswamiUV07} only achieve left-degree $\Omega(\log M\log\log\log M)$.\footnote{See Theorem 3.5 of \cite{GuruswamiUV07}, with $N=2^{O(n/\log\log n)}$, $K_{\max}=\sqrt{\log n}$, and $\eps < \frac12$.} 

Instead, our strategy is to pick a suitable $q$ and take \emph{all} $q$-sparse combinations, which directly ensures completeness as shown in the previous section. This requires us to pick an asymptotically smaller $q=O(\frac{n}{\log n\log\log n})$ to ensure we have at most $2^{O(n/\log \log n)}$ total equations. This contains many small (length $3$ or $4$) even covers, so to output a nontrivial upper bound on the value with high probability, we randomly partition the synthetic equations into parts of size $2^{Cn/\log\log n}$ for carefully chosen $C$. This guarantees that a $1-o(1)$ fraction of the parts will have no trivial even covers, and so we can use our aforementioned refutation algorithm to refute each part separately, obtaining an $\frac12+O(\eps)$ bound on all but an $O(\eps)$-fraction of parts, which is sufficient.
\subsection{Sum of squares lower bounds for linear systems}
To show that sum-of-squares hierarchy does not recognize our algorithm, we show that with high probability over a $O(\frac{n}{\log n})$-sparse linear system, there is a degree-$\Omega(n)$ pseudoexpectation which ``thinks'' the system is satisfiable. We use a result of \cite{Sch08} showing that such a pseudoexpectation exists if there is no width-$w=\Omega(n)$ resolution proof of unsatisfiability. That is, viewing the system as monomial equations over $\{\pm 1\}$, we cannot derive $0=1$ by multiplying pairs of equations together to create new equations without generating a degree $w$ monomial. We then show that with high probability, a random linear system contains no short even cover, so any resolution refutation must contain a large number of equations, and moreover there exists a interval $I=[t,2t]$ such that any combination of $s\in I$ total $a_j$'s has Hamming weight $\Omega(n)$. Hence, it is not possible to generate an even cover through resolution without generating a degree-$\Omega(n)$ monomial. We present the proof in \Cref{sec:sos}.
\section{Refuting Semirandom Linear Systems}
\label{sec:semirandom}
In this section, we prove our main algorithmic results, \cref{thm:semirandomfull,thm:semirandomgeneral}, which we restate below.
\semirandomfull*
\semirandomgeneral*

We prove these results by introducing two closely related algorithms, both of which rely on a subroutine which takes in a semirandom linear system derived from a system $(A,b)$ and outputs an upper bound on its value. We begin by formally defining the notion of a system induced by a set of index sets.
\begin{definition}
    Let $A\in \F_2^{m\times n}$, $b\in \F_2^m$, and $S\subseteq 2^{[m]}$ be a multiset of subsets of $[m]$. Then we call the system of equations $\{\ip{\sum_{t\in T}a_t,x}=\sum_{t\in T}b_t\}_{T\in S}$ the $S$-induced linear system of $(A,b)$. We define $\val(A,b,S)$ to be the value of the $S$-induced linear system of $(A,b)$, i.e. the maximum over all $x\in \F_2^n$ of the fraction of equations satisfied by $x$. We define $\bias(A,b,S)\coloneqq \val(A,b,S)-\frac12$.
\end{definition}
\newcommand{\cT}{\mathcal{T}}
The subroutine, which we will present in \Cref{sec:semirandomlemmas}, takes in a semirandom linear system $(A,b)$ with $m$ total equations, parameters $M,r$, and an $M$-tuple $\cT$ of index sets, and operates on the $\cT$-induced linear system of $(A,b)$. The parameter $r$ describes the number of bucketing steps our algorithm will take where we combine pairs of equations. When $m\geq 2^{O(n/\log n)}$, the inputs $M,\cT$ are redundant, and we simply run the subroutine on the original linear system $(A,b)$ (here, we have $M=m$ and $\cT=(\{1\}, \cdots, \{m\})$); however, we will use induced systems in the regime where $m=n^{1+\gamma}\ll 2^{O(n/\log n)}$.
\begin{lemma}\label{lem:subroutineproperties}
    There exists an algorithm REFUTE taking in inputs $M, A\in \F_2^{m\times n}$, $b\in \F_2^m$, $r\geq 4$, $\cT\in (2^{[m]})^M$ and outputs some $\cB\in [0,\frac12]$ such that: 
    \begin{enumerate}
        \item The runtime of the algorithm is $r\poly(M2^{n/r})$ time. \label{item:runtime}
        \item $\val(A,b,\cT)\leq \frac12+\cB$. \label{item:valbound}
        \item If $M\geq 2^{4n/3r+2n/r^2+5r}$ and $\cT$ does not contain even covers of length at most $2^{r+1}$, then the algorithm terminates early and outputs $\frac12$ with probability at most $25\cdot 2^{-n/6r}$ over the randomness introduced by the algorithm.\label{item:termination}
        \item If $(A,b)$ is a semirandom linear system and conditioned on the algorithm not terminating prematurely, whenever $M\geq r2^{n/r+n/r^2}$, then for any $\eps>2^{-n/r^22^{r+1}}$, we have $\cB<\eps$ with probability at least $1-1/(3M\eps^{2^{r+1}}2^{n/6r})$.
        \label{item:epsbound}
    \end{enumerate}
\end{lemma}

We begin by proving \Cref{thm:semirandomfull}.
\begin{proof}[Proof of \Cref{thm:semirandomfull}]
    We present the following algorithm, which defines a value of $r$ and runs the REFUTE subroutine.
    \begin{mdframed}
    \begin{algorithm}[Refutation]
    \label{alg:refutation1}    
    \mbox{}
    \begin{description}
        \item[Input:] $a_1,\cdots,a_m\in\F_2^{ n}$, $b_1,\cdots, b_m\in \F_2$
        \item[Output:] $\cV\in [\frac12,1]$ such that at most $\cV$ fraction of equations $\ip{a_i,x}=b_i$ are satisfiable.
        \item[Operation:] \mbox{}
            \begin{enumerate}
                \item Define $r\coloneqq \frac{\log n}{2}$ and $s\coloneqq n/r$. Let $M\coloneqq m$ and $q\coloneqq 1$.
            
            Define the stacked matrix $A\in \F_2^{m\times n}$ of the $a_i$s, the vector $b\in \F_2^m$ of the $b_i$s, and $\cT=(\{1\},\{2\},\cdots, \{m\}))$.
            
            \item Output $\frac12+$REFUTE$(m, A, b, r, q,\cT)$.
            \end{enumerate}
    \end{description}
    \end{algorithm}
\end{mdframed}
    This algorithm calls the REFUTE subroutine with $M=m$ and $r=O(\log n)$, so we apply \Cref{lem:subroutineproperties}. By \Cref{item:runtime}, this runs in $\poly(M)=\poly(m)$ time, and the first condition of the theorem follows from \Cref{item:valbound}. By \Cref{item:termination}, since $\cT=(\{1\}, \{2\}, \cdots, \{m\})$ contains no even covers, the subroutine does not terminate early with high probability, so the second condition of the theorem then follows from \Cref{item:epsbound}.
\end{proof}
We then prove \Cref{thm:semirandomgeneral}. We first state the following result relating the value of a linear system $(A,b)$ to the value of its ${[m]\choose q}$-induced linear system.
\begin{restatable}{lemma}{fullinduced}\label{lem:fullinduced}
    Let $(A,b)$ be a linear system in $\F_2^n$ with $m$ equations. If $q\leq \eps^2 m$ and $\val(A,b,{[m]\choose q})\leq\frac12+\eps^q$, then $\val(A,b)\leq\frac12+\eps$.
\end{restatable}
We also bound the number of parts containing short even covers when we randomly partition the set of all size-$q=O(n/\gamma\log n\log\log n)$ subsets of a size-$n^{1+\gamma}$ size into size-$2^{O(n/\log\log n)}$ parts.
\begin{restatable}{lemma}{noevencovers}
    \label{lem:noevencovers}
    Let $m=n^{1+\gamma}$, $q=\frac{6n}{\gamma\log n\log\log n}$, and let $\cH$ be the complete $q$-uniform hypergraph on vertex set $[m]$. Consider a uniformly random partition of $\cH$ into subhypergraphs with $2^{3n/\log\log n}$ edges each. Then with probability $1-2^{-n/\gamma\log n}$, at most a $2^{-8n/\gamma\log n}$ fraction of the subhypergraphs contain an even cover of size at most $\log n$.
\end{restatable}
We defer the proof of these two lemmas to \Cref{sec:noevencovers}. Using them, we can prove the theorem.
\begin{proof}[Proof of \Cref{thm:semirandomgeneral}]
We present the following algorithm. This algorithm generates equations corresponding to all size-$q$ subsets of the equations, then uses the REFUTE subroutine to refute each part.
\begin{mdframed}
    \begin{algorithm}[Refutation]
    \label{alg:refutation2}    
    \mbox{}
    \begin{description}
        \item[Input:] $a_1,\cdots,a_m\in\F_2^{ n}$, $b_1,\cdots, b_m\in \F_2$
        \item[Output:] $\cV\in [\frac12,1]$ such that at most $\cV$ fraction of equations $\ip{a_i,x}=b_i$ are satisfiable.
        \item[Operation:] \mbox{}

            \begin{enumerate}
                \item Let $r\coloneqq\frac{\log\log n}{2}$, $s\coloneqq n/r$, $M\coloneqq 2^{3n/\log\log n}$, $q\coloneqq 6\left(\frac{n}{\log(m/n)\log\log n}\right)$, $p={m\choose q}/M$.
                \\
                Define the stacked matrix $A\in \F_2^{m\times n}$ of the $a_i$s and the vector $b\in \F_2^m$ of the $b_i$s
                \item Randomly partition the set of all size-$q$ subsets of $[m]$ into $p$ sets $\cF_1, \cdots, \cF_p$ of size $M$.
                \item Initialize $\delta\coloneqq 0$.
                \item For all $1\leq j\leq p$ :
                \begin{itemize}
                    \item Enumerate the elements of $\cF_j$ as $T_1, \cdots, T_M$ and write $\cT_j=(T_1, \cdots, T_M)$.
            \item Update $\delta\leftarrow \delta + \frac1p$REFUTE$(M, A, b, r,q,\cT_j)$.                \end{itemize}
                \item Output $\frac12+\min(\frac12,\delta^{1/q})$.
            \end{enumerate}                
    \end{description}
    \end{algorithm}
\end{mdframed}
This algorithm calls the subroutine \Cref{alg:semirandomrefute} a total of $p=2^{O(n/\log\log n)}$ times with $r=\frac12\log\log n$. We apply \Cref{lem:subroutineproperties}. By \Cref{item:runtime}, the total runtime for all $p$ calls is $p\cdot \poly(M)=2^{O(n/\log\log n)}$ time. By \Cref{item:valbound} applied to each call to the subroutine, we have $\val(A,b,{[m]\choose q})\leq\frac12+\delta$, so by \Cref{lem:fullinduced}, $\val(A,b)\leq \frac12+\delta^{1/q}$, proving the first condition.
By \Cref{lem:noevencovers}, with high probability, at most $2^{-8n/\gamma\log n}\leq\frac{\eps^q}{3}$ fraction of calls will contain a length $\leq 2^{r+1}$ even cover, in which case the subroutine may output $\frac12$. Otherwise, by \Cref{item:termination} and a Chernoff bound, with high probability, at most a $\frac{\eps^q}{3}$ fraction of the remaining calls to the subroutine will terminate early and output $\frac12$. By a union bound, with probability at least $1-{m\choose q}/(3M^2(\eps^q/3)^{2^{r+1}}2^{n/6r})=1-o(1)$, none of the remaining at most ${m\choose q}/M$ calls which do not terminate early output some value $\cB'\geq\frac{\eps^q}{3}$. Therefore, with high probability, we have $\delta\leq \frac{\eps^q}{3}+\frac{\eps^q}{3}+p\cdot\frac{\eps^q}{3p}=\eps^q$ and thus the second condition follows.     
\end{proof}

\subsection{Proof of \Cref{lem:subroutineproperties}} \label{sec:semirandomlemmas}

We present the following algorithm, which takes in a linear system $(A,b)$, a set of index sets assumed to be of size $q$, and a parameter $r$. It operates a set of synthetic equations formed by summing up the equations of $(A,b)$ indexed by the provided index sets.
The algorithm, given a $\cT$-induced system, zeroes out blocks of size $n/r$ one by one, creating a new induced system at each step, until after $r$ steps the $LHS$ of each equation is the zero vector. This allows us to determine an upper bound on the value of this final induced system, which implies an upper bound on $\val(A,b)$ as we will prove later. During the operation of the algorithm, $f_k$ keeps track of the RHS of combined equations, $g_k$ keeps track of which of original equations contribute to a combined equation, and $h_k$ which equations of the input to the subroutine contribute to a combined equation.
\begin{mdframed}
    \begin{algorithm}[Refute]
    \label{alg:semirandomrefute}    
    \mbox{}
    \begin{description}
        \item[Input:] $M$, $A\in\F_2^{m\times n}$, $b\in \F_2^m$, $r \in \N$, $q \in \N$, $\cT\in ({2^{[m]}})^M$
        \item[Output:] $\eps\in [0,\frac12]$
        \item[Operation:] \mbox{}
            \begin{enumerate}[(1)]
            \item If $M<50\cdot 2^{7n/6r+n/r^2}$, terminate and output $\frac12$.
            \item Define $s=n/r$ and $S_0\coloneqq \{T_i:i\in [M]\}$. Define functions $f_0:2^{[m]}\mapsto \F_2$, $g_0:2^{[m]}\mapsto 2^{[m]}$, and $h_0:2^{[m]}\mapsto 2^{[M]}$ such that for all $i\in [M]$, $f_0(T_i)=\sum_{t\in T_i}b_t$, $g_0(T_i)=T_i$, and $h_0(T_i)=\{i\}$.
            \item For $k=0, \cdots, r-1$:
                \begin{enumerate}
                    \item For all $v\in \F_2^{s}$, define the multiset $B_{k,v}\coloneqq\{T\in S_k:(\sum_{t\in T} a_t)_{[n-(k+1)s+1,n-ks]}=v\}$. If $|B_{k,v}|$ is not a multiple of $2^{n/r^2}$, delete up to $2^{n/r^2}-1$ elements of $B_{k,v}$ to make it a multiple of $2^{n/r^2}$.
                    \label{step:formbkvs}
                    \item Randomly partition $B_{k,v}$ into $t_{k,v}\coloneqq|B_{k,v}|\cdot 2^{-n/r^2}$ sets $B_{k,v}^{(1)},\cdots, B_{k,v}^{(t_{k,v})}$
                    \item For all $v\in \F_2^s$ and $j\in [t_v]$, define the multiset $B_{k,v}^{*(j)}= \{T^{(1)}\oplus T^{(2)}:T^{(1)}\neq T^{(2)}\in B_{k,v}^{(j)}\}$. For all $T^{(1)}\neq T^{(2)}\in B_{k,v}^{(j)}$, set:
                    \begin{itemize}
                        \item $f_{k+1}(T^{(1)}\oplus T^{(2)})=f_k(T^{(1)})+f_k(T^{(2)})$,
                        \item $g_{k+1}(T^{(1)}\oplus T^{(2)})=g_k(T^{(1)})\oplus g_k(T^{(2)})$,
                        \item $h_{k+1}(T^{(1)}\oplus T^{(2)})=h_k(T^{(1)})\oplus h_k(T^{(2)})$.
                        \item If $g_{k+1}(T^{(1)}\oplus T^{(2)})=\{\}$ and $h_{k+1}(T^{(1)}\oplus T^{(2)})\neq \{\}$, terminate and output $\frac12$.\label{step:elimecs}
                    \end{itemize} \label{step:combine}
                    \item $S_{k+1}\leftarrow\bigcup_{v\in \F_2^s}\cup_{i\in [t_v]}B_v^{*(i)}$
                    \item Delete all $a\in S_{k+1}$ such that $h_{k+1}(a)=\{\}$. If we delete more than a $2^{-n/6r}$ fraction of $S_{k+1}$, terminate and output $\frac12$.\label{step:prune}
                \end{enumerate}
            \item Count the number of pairs $T'\neq T''\in S_r$ such that $h_r(T')=h_r(T'')$. If this is larger than a  $\frac{1}{M}\cdot 2^{-n/6r}$ fraction of pairs, terminate and output $\frac12$. \label{step:finalcount}
            \item If there exists a pair $T'\neq T''\in S_r$ such that $g_r(T')=g_r(T'')$ and $h_r(T')\neq h_r(T'')$, terminate and output $\frac12$.\label{step:finalelimec}
            \item Let $\Delta\in [-\frac12,\frac12]$ such that $f_r(a)=0$ for $\frac{1}{2}+\Delta$ fraction of $a\in S_r$. Output \\$\min(\frac12,\max(2^{-n/r^22^{r+1}},(\max(\Delta,0))^{1/2^r}))$.
        \end{enumerate}
    \end{description}
    \end{algorithm}
\end{mdframed}
We note that for any $0\leq k\leq r$, the set of equations $\sum_{t\in T}a_t = \sum_{t\in T}b_t$ for all $T\in S_k$ is the $S_k$-induced system of $(A,b)$.

\begin{remark}
    We assume that if some $T\in S_k$ appears multiple times in the multiset $S_k$, then each copy of $T$ can be identified, so $f_k,g_k,h_k$ are functions. In the description of the algorithm, $T^{(1)}\neq T^{(2)}$ indicates that $T^{(1)}$ and $T^{(2)}$ are separate elements of the multiset that may be the same set. 
\end{remark}
We will prove the following lemmas, which together prove \Cref{lem:subroutineproperties}.
\begin{restatable}{lemma}{runtime}
\label{lem:runtime}
    \Cref{alg:semirandomrefute} runs in $r\poly(M2^{n/r})$ time.
\end{restatable}
\begin{restatable}{lemma}{semirandomcondone}
\label{lem:semirandomcondone}
    Suppose $r\geq 4$ and let $\cB$ be the output of \Cref{alg:semirandomrefute}. Then $\val(A,b,\cT)\leq \frac12+\cB$.
\end{restatable}
\begin{lemma}\label{lem:termination}
If $M\geq 2^{4n/3r+2n/r^2+5r}$ and $\cT$ does not contain even covers of length at most $2^{r+1}$, then \Cref{alg:semirandomrefute} terminates early and outputs $\frac12$ with probability at most $25\cdot 2^{-n/6r}$ over the randomness introduced by the algorithm.
\end{lemma}
\begin{lemma}\label{lem:epsbound}
    If $(A,b)$ is a semirandom linear system and \Cref{alg:semirandomrefute} does not terminate prematurely and $M\geq r2^{n/r+n/r^2}$, then for any $\eps>2^{-n/r^22^{r+1}}$, $\cB<\eps$ with probability at least $1-1/(3M\eps^{2^{r+1}}2^{n/6r})$.
\end{lemma}
First, we have the following numerical bounds on the sizes of the $S_k$ in terms of $|S_0|=M$.
\begin{proposition}\label{prop:boringcalculation}
    For all $1\leq k\leq r$ such that \Cref{alg:semirandomrefute} does not terminate before finishing iteration $k$, if where $r\geq 4$ and $n\geq 4r^2$, $2^{kn/r^2-2k}(|S_0|-k2^{n/r+n/r^2})\leq |S_k|\leq 2^{kn/r^2-k}|S_0|$. Furthermore, if $M\geq r2^{n/r+n/r^2}$, then $|S_{k-1}|\leq 4\cdot 2^{-n/r^2}|S_k|$. 
\end{proposition}
\begin{proof}[Proof of \Cref{prop:boringcalculation}]
    A direct calculation shows that $|S_k|\leq |S_0|2^{-kn/r^2}{2^{n/r^2}\choose 2}^k\leq \frac{1}{2^k}2^{kn/r^2}|S_0|$. We note that in \Cref{step:formbkvs}, we delete at most $2^{n/r+n/r^2}$ total elements of $S_{k-1}$, and in \Cref{step:prune}, we delete at most a $2^{-n/6r}$ fraction of $S_k$, so $|S_k|\geq (|S_{k-1}|-2^{n/r+n/r^2})2^{-n/r^2}{2^{n/r^2}\choose 2}(1-2^{-n/6r})\geq (|S_0|-k2^{n/r+n/r^2})2^{kn/r^2}/4^k$.
\end{proof}
We use this to bound the runtime.
\begin{proof}[Proof of \Cref{lem:runtime}]
    At each of the $r$ steps of the subroutine, we have $|S_k|\leq |S_r|\leq M2^{2n/r}$ by \Cref{prop:boringcalculation}, and each step in iteration $k$ takes at most $\poly(|S_k|)$ time, for a total runtime of $r\poly(M2^{n/r})$.
\end{proof}
Next, we lower bound the value of the $S_{k+1}$-induced linear system as a function of the value of the $S_k$-induced linear system.
\begin{lemma}\label{lem:eachiteration}
    Let $0\leq k\leq r-1$ such that $2^{n/r+n/r^2}\leq 2^{-n/6r}M/50$. 
     If $\bias(A,b,S_k)\geq 2^{-n/4r^2}$, then $\bias(A,b,S_{k+1})\geq \bias(A,b,S_k)^2$.
\end{lemma}
\begin{proof}[Proof of \Cref{lem:eachiteration}]
    Let $\gamma=\bias(A,b,S_k)$ and pick some $x\in \F_2^n$ such that $\ip{\sum_{t\in T}a_t,x}=f_k(T)$ for exactly a $(\frac12+\gamma)$ fraction of all $T\in S_k$. First, note that in \Cref{step:formbkvs}, we delete at most $2^{n/r+n/r^2}\leq 2^{-n/6r}M/50\leq \gamma M/50$ total elements of $S_k$. Let $S_k'$ be the set of all remaining elements after deletion.  For all $v\in \F_2^s$ and $j\in [t_{k,v}]$, let $c_{j,v}$ be the total number of $T\in \bjk$ such that $\ip{\sum_{t\in T}a_t,x}=f_k(T)$, and observe that $\sum_{v}\sum_j c_{j,v}\geq (\frac12+0.99\gamma)|S_k'|$.
    
    Then for all $T^{(1)}\neq T^{(2)}\in \bjk$, we have $\ip{\sum_{t\in T^{(1)}\oplus T^{(2)}}a_t,x}=f_{k+1}(T^{(1)}+T^{(2)})=f_k(T^{(1)})+f_k(T^{(2)})$ if and only if both $\ip{\sum_{t\in T^{(1)}}a_t,x}=f_k(T^{(1)})$ and $\ip{\sum_{t\in T^{(2)}}a_t,x}=f_k(T^{(2)})$, or both $\ip{\sum_{t\in T^{(1)}}a_t,x}\neq f_k(T^{(1)})$ and $\ip{\sum_{t\in T^{(2)}}a_t,x}\neq f_k(T^{(2)})$. Thus the total number of $T'\in S_{k+1}$ such that $\ip{\sum_{t\in T'}a_t,x}=f_{k+1}(T')$ is, by Cauchy-Schwarz,
    \begin{align*}
        \sum_{\substack{v\in F_2^s \\ 1\leq j\leq t_{k,v}} }{c_{j,v}\choose 2}
        &+ {2^{n/r^2}-c_{j,v}\choose 2}
        \\&=  \sum_{\substack{v\in F_2^s \\ 1\leq j\leq t_{k,v}} } c_{j,v}^2 - 2^{n/r^2}\sum_{\substack{v\in F_2^s \\ 1\leq j\leq t_{k,v}} } c_{j,v}+\frac{|S_k'|2^{n/r^2}-|S_k'|}{2}
        \\&\geq \left(\frac12+0.99\gamma\right)^2|S'_k|^2\frac{2^{n/r^2}}{|S'_k|}- 2^{n/r^2}\left(\frac{1}{2}+0.99\gamma\right)|S'_k|+\frac{|S'_k|2^{n/r^2}-|S_k'|}{2}
        \\&\geq \left( \frac{1}{2}+1.5\gamma^2\right)\frac{|S'_k|}{2^{n/r^2}}{2^{n/r^2}\choose 2}
        =\left(\frac12+1.5\gamma^2\right)|S_{k+1}|\mper
    \end{align*}
    Finally, note that if \Cref{step:prune} does not terminate, we delete only a $2^{-n/6r}$-fraction of $S_{k+1}$, so a $\geq\frac12+\gamma^2$ fraction of all $a'\in S_{k+1}$ post-pruning satisfy $\ip{\sum_{t\in T'}a_t,x}=f_{k+1}(T')$. Thus $\bias(A,b,S_{k+1})\geq \bias(A,b,S_k)^2$.
\end{proof}
Using this, we prove \Cref{lem:semirandomcondone}.
\begin{proof}[Proof of \Cref{lem:semirandomcondone}]
    If \Cref{alg:semirandomrefute} terminates early and outputs $\frac12$, then $\val(A,b,\cT)\leq 1=\frac12+\cB$. If $\bias(A,b,\cT)\leq 2^{-n/r^22^{r+1}}$, then $\val(A,b,\cT)\leq \frac12+2^{-n/r^22^{r+1}}=\frac12+\cB$. Otherwise, at the last step of \Cref{alg:semirandomrefute}, we have $\sum_{t\in T}a_t=\mathbf0$ for all $T\in S_r$, so for all $x\in\F_2^n$, $\ip{\sum_{t\in T}a_t,x}=f_r(T)$ if and only if $f_r(T)=0$.
    Therefore, repeatedly applying \Cref{lem:eachiteration}, we have $\frac12+\cB\geq \frac{1}{2}+\bias(A,b,S_r)^{1/2^r}\geq \frac{1}{2}+(\bias(A,b,S_0)^{2^r})^{1/2^r}=\val(A,b,S_0)=\val(A,b,\cT)$.
\end{proof}
We upper bound the number of $T\in S_k$ whose derivation includes some $i\in [M]$ an odd number of times, which we use as a crude upper bound on the number of times a single synthetic equation can appear in $S_k$.
\begin{lemma}\label{lem:gkibound}
    For all $i\in [M]$ and $0\leq k\leq r$, there are at most $2^{nk/r^2}$ total $T\in S_k$ such that $i\in h_k(T)$.
\end{lemma}
\begin{proof}[Proof of \Cref{lem:gkibound}]
    We use induction. When $k=0$, there is exactly one $T=T_i\in S_0$ such that $i\in h_0(T_i)=\{i\}$. Now suppose the result holds up to $k$. Let $T\in S_{k+1}$ such that $i\in h_{k+1}(T)$, and suppose $T$ is derived from $T=T'\oplus T''$ for $T',T''\in S_k$, and exactly one of $h_k(T')_i,h_k(T'')_i$ contains $i$. Without loss of generality, let that be $T'$. Then note that there are at most $2^{nk/r^2}$ total $T'$ such that $i\in h_k(T')$, and there are $2^{n/r^2}-1<2^{n/r^2}$
    other $T''$ in the same $\bjk$ as $T'$, so there are at most $2^{nk/r^2}\cdot 2^{n/r^2}=2^{n(k+1)/r^2}$ ways to pick $T'$ and $T''$ to make $T$. This completes the inductive step and we are done by induction.
\end{proof}
Using this, we can bound the probability that the subroutine terminates early (and outputs $\frac12$).
\begin{lemma}\label{lem:fewdeletions}
    If $M\geq 2^{4n/3r+2r}$, then with probability $1-2^{-n/6r}$, \Cref{alg:semirandomrefute} does not terminate prematurely due to \Cref{step:prune}.
\end{lemma}
\begin{proof}[Proof of \Cref{lem:fewdeletions}]
    Let $0\leq k\leq r-1$. We first bound the probability that the algorithm terminates at \Cref{step:prune} in iteration $k$. Suppose that $h_{k+1}(T)=\{\}$ for some $T\in S_{k+1}$. Then suppose $T$ is constructed by $T=T'+T''$ for $T',T''\in S_{k}$. Then note that $h_{k+1}(T)=h_k(T')+h_k(T'')$ and $h_k(T'),h_k(T'')\neq \{\}$, so we must have $h_k(T')=h_k(T'')\neq \{\}$. Then $h_k(T')$ must contain a nonzero coordinate, and by \Cref{lem:gkibound}, this coordinate is nonzero for at most $2^{nk/r^2}$ total $T'\in S_k$. Therefore every value of $h_k(T')$ can occur at most $2^{nk/r^2}$ times.

    Now, we will bound the fraction of pairs $T',T''$ in the same $\bjk$ such that $h_k(T')=h_k(T'')$. Let $X$ be the set of all unique $h_k(T')$, and for all $v\in \F_2^r, x\in X$ let $\wvx$ denote the number of $T\in B_{k,v}$ such that $h_k(T)=x$, so $\wvx\leq 2^{kn/r^2}$. Fix some $v\in \F_2^r$. Then we have $\sum_{x\in X}\wvx = |B_{k,v}|=2^{n/r^2}t_{k,v}$. Furthermore, note that the total number of $T',T''\in B_{k,v}$ such that $h_k(T')=h_k(T'')$ is
    \begin{align*}
        \sum_{x\in X}{\wvx\choose 2}\leq \sum_{x\in X}\frac{2^{kn/r^2}\wvx}{2}=2^{kn/r^2-1}\cdot |B_{k,v}|\mcom
    \end{align*}
    and each such pair ends up in the same $\bjk$ with probability $\frac{2^{n/r^2}-1}{|B_{k,v}|-1}$ when we do random partitioning. Therefore, the expected number of pairs with equal values of $h_k(\cdot)$ and which lie in the same $\bjk$ is 
    \begin{align*}
        \sum_{v\in \F_2^s}\frac{2^{n/r^2}-1}{|B_{k,v}|-1}\cdot 2^{kn/r^2-1}|B_{k,v}|
        \leq \sum_{v\in \F_2^s}\frac{2^{n/r^2}}{|B_{k,v}|}\cdot 2^{kn/r^2-1}|B_{k,v}|=2^{n/r}\cdot 2^{(k+1)n/r^2-1}\leq \frac{2^{n/r+2k+2}}{M}|S_{k+1}|\mcom
    \end{align*}
    where the last inequality uses \Cref{prop:boringcalculation} and the lower bound on $M$. Therefore, by Markov, the probability that we delete more than a $2^{-n/6r}$ fraction of elements of $S_{k+1}$ is
    \begin{align*}
    \frac{2^{n/r+(k+1)n/r^2-1}}{|S_{k+1}|\cdot 2^{-n/6r}}
    &\leq \frac{3\cdot 2^{n/r+(k+1)n/r^2-1+n/6r}}{2^{(k+1)n/r^2-2(k+1)}\left(|S_0|-(k+1)2^{n/r+n/r^2}\right)}
    \\&\leq O(1)\cdot\frac{2^{7n/6r+2k+1}}{|S_0|}
    \leq O(1)\cdot\frac{2^{7n/6r+2r}}{M}\mper
    \end{align*}
    Union bounding over all $r$ values of $k$ gives the desired result.
\end{proof}
\begin{lemma}\label{lem:fewdeletions2}
    If $M\geq 2^{4n/3r+2n/r^2+5r}$, then with probability $1-24\cdot 2^{-n/6r}$, \Cref{alg:semirandomrefute} does not terminate prematurely due to \Cref{step:finalcount}.
\end{lemma}
\newcommand{\aone}{a^{(1)}}
\newcommand{\atwo}{a^{(2)}}
\newcommand{\brk}{B_{r-1,v}^{(j)}}
\newcommand{\brp}{B_{r-1,v}^{(j')}}
\begin{proof}[Proof of \Cref{lem:fewdeletions2}]    
    We will bound the expected number of pairs $T',T''\in S_r$ such that $h_r(T')=h_r(T'')$. We begin by conditioning on $S_{r-1}$ and considering $T^{(1)},T^{(2)}\in S_{r-1}$ such that $T=T^{(1)}\oplus T^{(2)}$. Then note that $T\in S_{r}$ cannot occur unless $T^{(1)},T^{(2)}$ lie in the same pre-partitioned bucket, i.e. $T^{(1)},T^{(2)}\in B_{r-1,v}$ for some $v\in \F_2^s$. If they do lie in the same $B_{r-1,v}$, then we have $T\in S_r$ if $T^{(1)},T^{(2)}$ lie in the same $\brk$, which occur with probability at most $2^{n/r^2}/|B_{r-1,v}|$.

    Now, let us bound the number of $T^{(1)},T^{(2)},T^{(3)},T^{(4)}\in S_{r-1}$ where $\{T^{(1)},T^{(2)}\}\neq \{T^{(3)},T^{(4)}\}$ such that $h_r(T^{(1)}\oplus T^{(2)})=h_r(T^{(3)}\oplus T^{(4)})\neq \emptyset$. First, define $z_{v,i}$ to be the number of $T\in B_{r-1,v}$ such that $i\in h_{r-1}(T)$. Then we can bound the number of valid $T^{(1)},T^{(2)},T^{(3)},T^{(4)}$ by first picking some $i\in U\coloneqq h_r(T^{(1)}\oplus T^{(2)})=h_r(T^{(3)}\oplus T^{(4)})$, and noting that this implies that $i$ is in exactly one of $h_{r-1}(T^{(1)}),h_{r-1}(T^{(2)})$ and exactly one of $h_{r-1}(T^{(3)}),h_{r-1}(T^{(4)})$. There are $4$ ways to choose which two contain $i$, and without loss of generality, let $i\in h_{r-1}(T^{(1)}),h_{r-1}(T^{(3)})$. Then we can pick $u,v$ such that $T^{(1)}\in B_{r-1,u}$, $T^{(3)}\in B_{r-1,v}$, in which case there are $z_{v,i}$ and $z_{u,i}$ ways to pick $T^{(1)}$ and $T^{(3)}$, respectively. 
    
    Now conditioning on the value of $j$ such that $T^{(1)}\in \brk$, there are $2^{n/r^2}-1<2^{n/r^2}$ choices for $T^{(2)}$ in $\brk$. Then it remains to pick $T^{(4)}$. But note that $h_{r-1}(T^{(4)})=h_{r-1}(T^{(1)})\oplus h_{r-1}(T^{(2)})\oplus h_{r-1}(T^{(3)})\coloneqq W$, and either $W=\emptyset$, in which case there is no valid $T^{(4)}$, or there exists some $i\in W$, and thus by \Cref{lem:gkibound}, there are at most $2^{n(r-1)/r^2}\leq 2^{n/r}$ potential choices for $T^{(4)}$, the probability that each lies in the same $B^{(j')}_{k,u}$ as $T^{(3)}$ is at most $2^{n/r^2}/|B_{k,u}|$. Therefore, the expected number of $T^{(1)},T^{(2)},T^{(3)},T^{(4)}$ such that $h_r(T^{(1)}\oplus T^{(2)})=h_r(T^{(3)}\oplus T^{(4)})\neq \emptyset$ is at most
    \begin{align*}
        4\sum_{i}\sum_{u,v}z_{v,i}z_{u,i}\cdot 2^{n/r^2+n/r}\frac{2^{n/r^2}}{|B_{r-1,v}|}
        &\leq 2^{2+2n/r^2+n/r}\sum_{v}\sum_i\sum_{u}z_{u,i}\frac{z_{v,i}}{|B_{r-1,v}|}
        \\&\leq 2^{2+2n/r^2+n/r}\sum_v\frac{1}{|B_{r-1,v}|}\sum_i 2^{n/r}z_{v,i}
        \\&\leq 2^{2+2n/r^2+2n/r}\sum_v\frac{1}{|B_{r-1,v}|}2^{r-1}|B_{r-1,v}|
        \\&\leq 2^{1+r+2n/r^2+3n/r}
        \\&\leq \frac{2^{3+5r+2n/r^2+n/r}}{M}\cdot \frac{|S_r|^2}{M}\mcom
    \end{align*}
    so by Markov, \Cref{step:finalcount} fails with probability at most $24\cdot 2^{-n/6r}$. 
\end{proof}
Now we can prove \Cref{lem:termination}.
\begin{proof}[Proof of \Cref{lem:termination}]
    Note that if $\cT$ does not contain any length $\leq 2^{r+1}$, then \Cref{step:elimecs} and \Cref{step:finalelimec} will never cause the algorithm to terminate, as every $T\in S_k$ is the direct sum of at most $2^{k+1}$ distinct $T_i$s. By \Cref{lem:fewdeletions} and \Cref{lem:fewdeletions2}, the algorithm does not terminate at an iteration of \Cref{step:prune} or \Cref{step:finalcount} with probability $1-25\cdot 2^{-n/6r}$.
\end{proof}
Finally, we prove \Cref{lem:epsbound}.
\begin{proof}[Proof of \Cref{lem:epsbound}]
    If \Cref{alg:semirandomrefute} does not terminate prematurely, at the final step, we have $g_r(T)\neq \{\}$ for all $T\in S_r$. Therefore, for all $T\in S_r$, $f(T)=\sum_{i\in g_r(T)}b_i$ is a Bernoulli$(\frac12)$ random variable. Furthermore, $f_r(T)$ is pairwise independent of all $f_r(T')$ such that $g_r(T)\neq g_r(T')$. Since $g_r(T)=g_r(T')$ only when $h_r(T)= h_r(T')$ due to passing \Cref{step:finalelimec}, letting $X=\sum_{T\in S_r}(1-f(T))$ be the number of $T\in S_r$ such that $f(T)=0$, we have
\begin{align*}
    \Var\left[X\right]
    &=\E\left[\sum_{T,T'\in S_r}\left(f_r(T)-\frac12\right)\left(f_r(T')-\frac12\right)\right]
    \\&=\sum_{\substack{T, T'\in S_r\\ g_r(T)= g_r(T')}}\E\left[\left(f_r(T)-\frac12\right)^2\right]
    \leq \frac{|S_r|(1+2^{-n/6r}|S_r|/M)}{4}\leq \frac{|S_r|^2}{3M2^{n/6r}} \mcom
\end{align*}
where \Cref{prop:boringcalculation} guarantees that $|S_r|\geq 3M2^{n/6r}$. Thus, by Chebyshev, we have 
    \begin{align*}
        \Pr\left[X\geq \left( \frac12+\eps^{2^r}\right)|S_r|\right]
        \leq \frac{\Var(X)}{\eps^{2^{r+1}}|S_r|^2}\leq \frac{1}{3M\eps^{2^{r+1}}2^{n/6r}}\mper
    \end{align*}
\end{proof}
\subsection{Proofs of \Cref{lem:fullinduced,lem:noevencovers}}\label{sec:noevencovers}
Here we prove the lemmas needed to prove \Cref{thm:semirandomgeneral}. First, we show the bound on $\val(A,b)$ in terms of $\val(A,b,{[m]\choose q})$, which we restate below.
\fullinduced*
\begin{proof}
    We prove the contrapositive. Suppose there exists some $x\in\F_2^n$ such that $\ip{a_t,x}=b_t$ for $\geq\frac12+\eps$ fraction of $t\in [m]$. We will show that for all $q\leq\eps^2m$, $\ip{\sum_{t\in T}a_t,x}=\sum_{t\in T}b_t$ for $\geq\frac12+\eps^{q}$ fraction of $T\in {[m]\choose q}$. For all $t\in [m]$, let $y_t=1$ if $\ip{a_t,x}=b_t$ and $-1$ otherwise, and for any $T\subseteq [m]$, let $y_T=\prod_{t\in T}y_t$. First we show by induction that ${m\choose q}^{-1}\sum_{T\in {[m]\choose q}}y_T\geq\eps{m\choose q-1}^{-1}\sum_{T\in {[m]\choose q-1}}y_T\geq 2\eps^q$. Note that this holds for $q=1$ by definition, and if it holds for $q$,
    \begin{align*}
        \frac{1}{{m\choose q+1}}\sum_{T\in {[m]\choose q+1}}y_T
        &=\frac{1}{{m\choose q+1}}\cdot\frac{1}{q+1}\sum_{T\in {[m]\choose q}}y_T\sum_{t\notin T}y_t
        \\&=\frac{1}{{m\choose q+1}}\cdot\frac{1}{q+1}\left(\left(\sum_{T\in {[m]\choose q}}y_T\right)\left(\sum_{t\in [m]}y_t\right)-(m-q+1)\sum_{T\in {[m]\choose q-1}}y_T\right)
        \\&\geq \frac{1}{{m\choose q+1}}\cdot\frac{1}{q+1}\left(2\eps m\sum_{T\in {[m]\choose q}}y_T-(m-q+1)\frac{{m\choose q-1}}{\eps{m\choose q}}\sum_{T\in {[m]\choose q}}y_T\right)
        \\&\geq \left(\frac{2\eps m - \frac{q}{\eps}}{m-q}\right)\frac{1}{{m\choose q}}\sum_{T\in {[m]\choose q}}y_T
        \geq \frac{\eps}{{m\choose q}}\sum_{T\in {[m]\choose q}}y_T\geq 2\eps^{q+1}\mcom
    \end{align*}
    where the last inequality follows from $\frac{q}{\eps}\leq \eps m$. Thus we have
    $\val(A,b,{[m]\choose q})\geq \frac12+\frac12{m\choose q}^{-1}\sum_{T\in {[m]\choose q}}y_T\geq \frac12+\eps^q$ for all $q\leq \eps^2 m$, as desired.
\end{proof}
Next, we show that with high probability, randomly partitioning the complete $q$-uniform hypergraph on $m$ vertices into $p$ parts for carefully chosen $q,p$ produces very few parts which contain short even covers. We restate the full lemma below.
\noevencovers*
\begin{proof}[Proof of \Cref{lem:noevencovers}]
    Fix an arbitrary subhypergraph of the partition. For any $3\leq t\leq \log n$, any even cover of size $t$ uses at most $tq/2$ vertices, so there are at most ${m\choose tq/2}{tq/2\choose q}^t$ potential size-$t$ even covers. The probability that each occurs is less than $(2^{3n/\log\log n}/{m\choose q})^t$, so union bounding over all $t$ and all even covers, the probability that at least one length $\leq \log n$ even cover is in our chosen subhypergraph is 
    \begin{align*}
        \sum_{t=3}^{\log n}{m\choose tq/2}\left(\frac{2^{3n/\log\log n}{tq/2\choose q}}{{m\choose q}}\right)^t
        &\leq\sum_{t=3}^{\log n}\left(2^{3n/\log\log n}\left(\frac{etq}{2m}\right)^{q/2}\right)^t
        \\&\leq \sum_{t=3}^{\log n} (2^{3n/\log\log n}2^{-3n/\log\log n -9n/\gamma t\log n-\log\log n/t})^t
        \\&\leq 2^{-9n/\gamma\log n}\mper
    \end{align*}
    Therefore, by Markov, with probability $1-2^{-n/\gamma\log n}$, at most a $2^{-8n/\gamma\log n}$ fraction of subhypergraphs contain an even cover of size $\leq \log n$.
\end{proof}


\section{Sum-of-Squares Lower Bounds for Refuting Linear Systems}
\label{sec:sos}
In this section, we prove \cref{thm:infsoslowerbound}. Before stating the formal theorem, we recall the standard definition of the degree-$d$ sum-of-squares hierarchy.
\begin{definition}
\label{sosaxioms}
    For any $d\geq 2$, a degree $d$ pseudo-expectation $\pE$ over $\Fits^n$ is a linear functional $\pE: \R[x_1, \ldots, x_n]_{\leq d}\to\R$ satisfying the following properties:
    \begin{enumerate}
        \item (Normalization) $\pE[1] = 1$, 
        \item (Booleanity) $\pE[fx_i^2] = \pE[f]$ for all $i\in[n]$ and $f\in\R[x_1, \ldots, x_n]_{\leq d - 2}$, 
        \item (Positivity) $\pE[f^2]\geq 0$ for all $f\in\R[x_1, \ldots, x_n]_{\leq d/2}$. 
    \end{enumerate}
\end{definition}

The formal theorem that we prove is \Cref{thm:soslowerbound}, stated below.
\begin{restatable}{theorem}{soslowerbound}
\label{thm:soslowerbound}
    Let $(A,b)$ be a linear system with $m=2^{Cn/\log n}$ equations such that each $a_1, \cdots, a_m$ is generated by independently setting each coordinate $(a_{i})_j=1$ with probability $q=\lambda\cdot\frac{8C\ln 2}{\log n}$ for some $1<\lambda\leq o(\log n)$ and $(a_i)_j=0$ with probability $1-q$. For all $1\leq i\leq m$, define the $p_i(x)=\prod_{j:(a_i)_j=1}x_j$. Then with high probability, there exists a degree $\Omega(n)$ pseudoexpectation $\pE$ such that $\pE p_i(x)=(-1)^{b_i}$ for all $i\in [m]$. 
\end{restatable}
Note that \cref{thm:soslowerbound} implies that the algorithm in \cref{thm:semirandomfull} is not captured by the sum-of-squares hierarchy, as the algorithm runs in $2^{O(n/\log n)}$ time, whereas by \cref{thm:soslowerbound}, sum-of-squares fails to refute the system of equations in $2^{\Omega(n)}$ time. We will use a result showing that a lower bound on the width of a resolution proof of unsatisfiability, which we define below, implies a lower bound on the SOS degree needed to refute the system.
\begin{definition}
    Given a system of multilinear polynomial equations $\{p_i(x)=b_i\}_{i\in [m]}$ over $\{\pm 1\}$, a \emph{resolution proof of unsatisfiability} is a sequence $(q_{m+1},c_{m+1}), (q_{m+2},c_{m+2}), \cdots (q_R,c_R)$, such that for all $i\geq m+1$, $q_i=q_jq_k$ and $c_i=c_jc_k$ for some $j,k<i$, ending in $(q_R,c_R)=(1,-1)$ or $(-1,1)$. The \emph{width} of a resolution proof of unsatisfiability is the maximum degree over all $q_{m+1}, \cdots, q_R$ after using the identity $x_i^2=1$ for all $i$ to reduce the $q$s to multilinear polynomials.
\end{definition}
\begin{fact}\cite[Lemma~13]{Sch08}\label{fact:sch}
    Let $(A,b)$ be a linear system with $m$ equations and for all $1\leq i\leq m$, define the polynomial $p_i(x)=\prod_{j:(a_i)_j=1}x_j$. Let $w>0$ such that the system $\{p_i(x)=(-1)^{b_i}\}_{i\in [m]}$ has no width-$w$ resolution proof of unsatisfiability. Then as long as $\max_i\deg(p_i)\leq w/2$, there exists a degree $w/2$ pseudoexpectation $\pE$ such that $\pE[p_i(x)]=(-1)^{b_i}$ for all $i\in[m]$.
\end{fact}
\begin{remark}
    \cite{Sch08} states the result in terms of $k$-XOR formulas where each clause has size exactly $k$; however, their proof never uses the fact that the clauses have equal or bounded size.
\end{remark}
In our proof, we utilize \Cref{fact:sch} by showing that any resolution derivation of a linear dependence in the set of $a_i$'s must contain a vector with Hamming weight $\Omega(n)$. We start by proving a technical lemma that shows that there is no short linear dependence in the $a_i$'s, and all $s$-combinations of $a_i$'s have large Hamming weight for all $s\in [t,2t]$ for some $t$.
\begin{lemma}\label{lem:evencovsizes}
    Suppose that $m=2^{Cn/\log n}$ and $p=  \lambda\cdot\frac{8C\ln 2}{\log n}$ for some $1<\lambda\leq o(\log n)$, and let $a_1, \cdots, a_m$ be independent random vectors in $\F_2^n$ where each coordinate of each $a_i$ is $1$ with probability $p$. Then there exist $t_1\geq t_2$ such that with high probability, all of the following hold:
    \begin{enumerate}
        \item For all $i\in [m]$, $a_i$ has Hamming weight $\leq 2pn$.
        \item For all $S_1\subseteq[m]$ with $0<|S_1|\leq t_1$, $\sum_{i\in S}a_i\neq \mathbf0$.
        \item For all $S_2\subseteq[m]$ with $\frac{t_2}{2}\leq |S_2|\leq t_2$, $\sum_{i\in S_2}a_i$ has Hamming weight $\Omega(n(1-1/\sqrt{\lambda}))$.
    \end{enumerate}
\end{lemma}
\begin{proof}

    The first condition follows directly from Chernoff and union bound.

    Let $t_1=\frac{2}{p}$. Consider $S_1=\{s_1,s_2, \cdots, s_{t}\}$ for some $0<t\leq t_1$. Then note that $\aio, \cdots, a_{s_t}$ are uniformly independent random vectors, and for all $i\in [t]$ and $j\in [n]$, the $j$th coordinate $(a_i)_j$ of $a_i$ is an independent Bernoulli$(p)$ random variable. Then for all $j\in [n]$ $\left(\sum_{i\in S_1}a_i\right)_j$ is an independent Bernoulli$\left(\frac{1-(1-2p)^{t_1}}{2}\right)$ random variable. Union bounding over all ${m\choose t}$ possible $S_1$ for all $0<t\leq t_1$, the probability that there exists some $S_1\subseteq[m]$ of size $t_1$ such that all coordinates of $\sum_{i\in S_1}a_i$ are 0 is at most
    \begin{align*}
        \sum_{t=1}^{t_1}{m\choose t}\left(\frac{1+(1-2p)^{t}}{2}\right)^n 
        &\leq \sum_{t=1}^{t_1}{m\choose t}\left(\frac{1+\exp(-2pt)}{2}\right)^n
        \\&\leq \sum_{t=1}^{t_1}\exp\left(\frac{Cnt\ln 2}{\log n}+n\ln(1+e^{-2pt}-n\ln 2)\right)
        \\&\leq \sum_{t=1}^{t_1}\exp(-npt/8)\leq t_1\exp(-np/8)=o(1)\mper
    \end{align*}
Let $t_2=\frac{1}{p}$ and $\gamma=1-\frac{1}{\sqrt{\lambda}}$. Consider $S_2\subseteq [m]$ for some $|S_2|=q$ with $\frac{t_2}{2}\leq q\leq t_2$. Then the Hamming weight of $\sum_{i\in S_2}A_i$ is a Binomial$\left(n, \frac{1-(1-2p)^{q}}{2}\right)$ random variable, and we have
\begin{align*}
    P\coloneqq\frac{1-(1-2p)^{q}}{2}\geq \frac12(1-\exp{(-2p(1/2p))})=\frac12-\frac{1}{2e}\geq 0.31\mcom
\end{align*}
so by a Chernoff Bound, the Hamming weight of $\sum_{i\in S_2}a_i$ is $\geq 0.31\gamma n=\Omega(n(1-1/\sqrt{\lambda}))$ with probability 
\begin{align*}
    1-\Pr\left[\text{Bin}(n,P)\leq 0.31\gamma n\right]\geq 1- \Pr\left[\text{Bin}(n,0.31)\leq 0.31\gamma n\right]
    &\geq 1-\exp\left(-n\cdot D_{KL}(0.31\gamma,0.31)\right)
    \\&\geq 1-\exp\left(-\frac{0.1922}{\lambda}n\right)\mper
\end{align*}
Union bounding over all possible $S_2$, we see that the third condition holds with probability
\begin{align*}
    1-\frac{t_2}{2}{m\choose t_2}\exp\left(-\frac{0.1922}{\lambda}n\right)
    &\geq 1-\exp\left(C\ln 2 \cdot\frac{1}{8\lambda C\ln 2}n+O(\log\log n)-\frac{0.1922}{\lambda}n\right)
    \\&\geq 1-\exp\left(-\Omega(n/\lambda)\right)=1-o(1)\mper
\end{align*}
\end{proof} 
We now prove \Cref{thm:soslowerbound}.

\begin{proof}[Proof of \Cref{thm:soslowerbound}]
    Consider an arbitrary resolution refutation of $\{p_i(x)=(-1)^{b_i}\}_{i\in [m]}$. This contains a sequence of monomials $\prod_{C\in I}x_C$ each with some index set $I$, one of which is some $I'\subseteq[m]$ such that $\prod_{C\in I'}x_C$ is identically $1$, which in turn corresponds to $\sum_{C\in I'}a_C=\mathbf0$. Furthermore, if the index sets are $I_1, \cdots, I_\ell=I'$, then we must have $I_{i_1}=I_{i_2}\oplus I_{i_3}$ for some $i_2,i_3<i_1$, and $|I_1|\leq |I_2|+|I_3|\leq 2\max(|I_2|,|I_3|)$. Thus for all $r\leq |I'|$, there exists some $s\leq \ell$ such that $|I_s|\in \left[\frac{r}{2},r\right]$. 
    
    Let $t_1,t_2$ be the constants from \Cref{lem:evencovsizes}. Then since $\sum_{C\in I'}a_C=\mathbf0$, we must have $|I'|\geq t_1>t_2$, so there exists some $I_r$ such that $|I_r|\in \left[\frac{t_2}{2},t_2\right]$, so $\sum_{C\in I_r}\mathbb{1}_C$ has Hamming weight $\Omega(n)$ and thus the monomial $\prod_{C\in I_r}x_C$ has degree $\Omega(n)$. Thus every resolution refutation of $\psi$ must have width $\Omega(n)$, so the result follows from \Cref{fact:sch}. 
\end{proof}

\section*{Acknowledgments}
We thank Avi Wigderson for helpful discussions and for participating in early stages of this research.

\bibliographystyle{alpha}
\bibliography{lpnr}

\end{document}